\documentclass[11pt,letterpaper]{article}

\usepackage{url}
\usepackage{fullpage}
\usepackage{amsmath}
\usepackage{algorithm}
\usepackage[noend]{algpseudocode}
\usepackage{amssymb}
\usepackage{amsthm}
\usepackage{color}
\usepackage{array}
\usepackage{xy}
\usepackage{setspace}
\usepackage{varwidth}% http://ctan.org/pkg/varwidth
\usepackage{multicol}
\usepackage{algorithm}
\usepackage{hyperref}
\usepackage{thm-restate}
\usepackage{todonotes}

 \newcommand{\E}{\mathbb{E}}

\newtheorem{thm}{Theorem}[section]

\theoremstyle{remark}

\newtheorem{theorem}{Theorem}[section]

\newtheorem{lemma}[thm]{Lemma}
\newtheorem{proposition}[thm]{Proposition}

\theoremstyle{remark}

\newcommand{\defn}[1]{\textbf{\emph{#1}}}
\renewcommand{\paragraph}[1]{\vspace{.2 cm} \noindent \textbf{#1}}

\usepackage{authblk}

\begin{document}

\title{With Leverage, the Gambler is Not the One Who Gets Ruined}   
\author{William Kuszmaul}
\affil{Carnegie Mellon University}
\date{}

\maketitle

\begin{abstract}
In this note we explore the following simple question: How does the classic \emph{gambler's ruin} phenomenon change if we allow the gambler \emph{financial leverage}? We show that, with a moderate amount of leverage, it is possible for an (idealized) gambler to sustain seemingly improbable profits over a remarkably long time horizon without ever going bankrupt, \emph{even} if the gambler makes only fair bets. With a leverage ratio of $L = 10$ (less than what Lehman Brothers Holdings had from 1995--2008), the probability of multiplying one's money by $1000\times$ without going bankrupt is roughly $50\%$. More generally, given $L$-fold leverage, and given the goal of multiplying one's money by a factor of $r > 1$, the optimal gambling strategy has probability roughly $r^{-1/L}$ of success. The results in this note are designed as part of the course development for the CMU randomized algorithms course.
\end{abstract}

\section{Introduction}

Consider the following simple model of a ``finance company'': The company begins with $n$ dollars, and wishes to turn it into $r \cdot n$ dollars for some $r > 1$. The company gets to make a sequence of bets, where each bet must have expected value $0$, and where no bet can risk more than the amount of money the company currently has. The company continues making bets until it either reaches its goal of at least $r \cdot n$ dollars, or until it goes bankrupt.

The classic Optional Stopping Theorem says that there is no way for this company to make an expected profit. Starting with $n$ dollars, the probability of reaching $r \cdot n$ before bankruptcy is at most $\frac{1}{r}.$ This basic insight, that it is very difficult for a (fair) gambler to make a significant profit without going bankrupt, is known as the \emph{gambler's ruin} \cite{Coolidge1909}. 

Now, suppose we see a real-world ``finance company'' that actually does turn $n$ into $r \cdot n$ for some large $r > 1$. Should we conclude that the ``bets'' being made by this company are actually profitable?

Consider, for example, Lehman Brothers Holdings, which from 1995 to 2008 multiplied their shareholder equity by a factor of roughly $8$, starting with around 3.4 billion dollars, and ending with around 26 billion dollars \cite{forstein2010mary}. The Optional Stopping Theorem suggests that, if they were making fair bets, they should achieve this with probability $\le \frac{3.4}{26} \approx \frac{1}{8}.$\footnote{It should be noted that this calculation ignores many factors. From year-end 1995 to year-end 2008, investing passively in the S\&P 500 would have resulted in a $\approx 1.8 \times$ growth. Additionally, in the case of Lehman, a potentially significant portion of the equity increase came not from investment returns, but from raising additional shareholder capital.} But there's a catch: When Lehman Brothers Holdings went bankrupt in 2008, it wasn't worth \emph{zero} dollars, it was worth roughly \emph{negative $150$ billion} \cite{cline2013lehman}. Lehman was leveraging their equity: Using the roughly $\$26$ billion they \emph{had}, they borrowed roughly $\$640$ billion to make bets with \cite{forstein2010mary}. In fact, for the entire period from 1995 to 2008, Lehman's leverage ratio (the ratio between total assets and shareholder equity) sat between $20$--$40\times$ \cite{forstein2010mary}. This allowed them to repeatedly make bets in which they risked more than they had.

We can capture this notion of \emph{leverage} as follows. Consider some parameter $L$, and suppose that our imaginary company is allowed to borrow up to $L - 1$ times its equity, at any given moment. So, if the company has $n$ \emph{actual} dollars, it can make a fair bet in which it risks up to $Ln$ dollars. We call this game, which we formalize in Section \ref{sec:main}, the $L$-Leverage Game. 

The main observation in this note is that, even with moderate leverage, the hypothetical company can make a \emph{remarkable} amount of profit before going bankrupt. Indeed, we show that there is a betting strategy which turns $n$ dollars into $r \cdot n$ dollars with probability arbitrarily close to $r^{-1/L}$. For perspective, this means that the ``thing that we can do with roughly 50\% probability, without going bankrupt'' is to multiply our money by a factor of $2^{L}$. 

We also prove a matching lower bound, saying that no strategy can achieve a probability of larger than $r^{-1/L}$ of obtaining $r \cdot n$ dollars before bankruptcy. This lower bound is also of pedagogical interest -- it offers a clean but nontrivial expository example of how to apply the Optional Stopping Theorem to reason about a betting game that does not fit within the standard ``fair-bets'' framework.

Returning to the Lehman example from earlier, from 1995 to 2008 Lehman multiplied their equity by a factor of roughly $r = 8$. During that same time period, they consistently had leverage ratio $L > 20$. Given this leverage ratio, our toy model says that an $r = 8 \times$ growth isn't that impressive. It has a success probability of roughly $8^{-1/20} \approx 90\%$, even with fair bets. In fact, with $20\times$ leverage, a hypothetical company can multiply its money by a factor of $2^{20} \approx 10^6$ with probability roughly $1/2$ of avoiding bankruptcy.

Real life, of course, is more complicated than a toy model. This note does not claim to capture the complexities of real-world finance (or the Lehman Brothers example, for that matter). Nonetheless, it does serve as an example of how, even with very simple probabilistic modeling, we can learn interesting things about the world. This note was written as part of course development for the CMU randomized algorithms course. 

\section{Formalizing and Analyzing the \texorpdfstring{$L$-}{L-}Leverage Game}\label{sec:main}

Throughout the rest of the paper, we will discuss our model as a single-player game, in which a player Alice starts with $1$ dollar, and makes a sequence of bets until she is either bankrupt (has $\le 0$ dollars) or has reached her goal of having $m \ge 1$ dollars. We call this the \defn{$L$-Leverage Game}. Note the slight change in variables from the version of the game discussed in the introduction; because the game is scale-invariant, we may as well assume that Alice starts with $1$ dollar (rather than $n$) and aims for $m$ dollars (rather than $r \cdot n$).

\paragraph{The $L$-Leverage Game. }Given a goal $m \ge 1$, and a leverage ratio $L \ge 1$, we will consider the following betting game. Let $X_0$ be the initial amount of money that Alice has, and let $X_t$, $t > 0$, be the amount of money that Alice has after the $t$-th step of the game. During the $t$-th step (and if the game has not yet ended), Alice gets to determine $X_t$ as follows: Based on the outcomes of $X_0, X_1, \ldots, X_{t - 1}$, Alice gets to select an arbitrary distribution for $X_t$ such that: 
\begin{itemize}
\item $\E[X_t | X_0, X_1, \ldots, X_{t - 1}] = X_{t - 1}$, and
\item $X_t \ge -(L - 1)X_{t - 1}$ with probability $1$.
\end{itemize}
Once Alice specifies the distribution, the random variable $X_t$ is sampled from that distribution. The game ends when $X_t \ge m$ or $X_t \le 0$. Alice is required to follow a betting strategy in which the game ends (eventually) with probability $1$. And, since there is no point in Alice placing a bet in which she ends up with strictly more than $m$ dollars, we will additionally assume throughout that Alice places bets in such a way that $X_t \le m$ for all $t$. 

Notice that the first condition on $X_t$ means that Alice must place bets with zero expected value. The sequence $X_0, X_1, \ldots, X_t$ is therefore a martingale. The second condition, on the other hand, captures the idea that Alice can, prior to each bet $t$, ``borrow'' up to $(L - 1)X_{t - 1}$ dollars, and place bets in which she stands to lose not just her $X_{t - 1}$ dollars, but also the $(L - 1)X_{t - 1}$ dollars that she borrowed. Hence Alice is able to risk up to an $L\times$ factor more than she has on each bet.

\paragraph{Warmup: $L = 1$. }As a warmup, let us consider the case where $L = 1$. This is just the classical leverage-free case, in which Alice must gamble using her own money. 

This setting can be analyzed using (a special case of) the Optional Stopping Theorem (OST): 
\begin{theorem}[Optional Stopping Theorem]
For $C > 0$, if $X_0, X_1, \ldots, X_t$ is a martingale satisfying $X_t \in [-C, C]$ for all $t$ with probability $1$, and if $\tau$ is a stopping time that exists with probability $1$, then $\E[X_\tau] = X_0$. If, instead, $X_0, X_1, \ldots$ is a supermartingale, then $\E[X_\tau] \le X_0$.
\end{theorem}

The OST tells us that, if Alice's game completes after the $\tau$-th step, then $\E[X_\tau] = X_0 = 1$. Since $X_{\tau}$ is either $0$ or $m$, it follows that $\Pr[X_\tau = m] = 1/m$. Furthermore, the conclusion is agnostic to Alice's strategy. In the case where Alice has no leverage, all strategies (that terminate with probability 1 and never allow for Alice to have $> m$ dollars) have the same winning probabilities as one another. 

\paragraph{How to play when $L > 1$. }Now, let us consider the case where $L > 1$. To motivate the optimal strategy for Alice, we first consider several suboptimal ones.

A natural but suboptimal strategy is to play as follows: Alice places a single bet in which she ends up with $m$ dollars with probability $L/(m + L - 1)$, and ends up with $-(L - 1)$ dollars with probability $(m - 1)/(m + L - 1)$. This bet results in $\E[X_1] = 1 = X_0$, so it is allowed, and it succeeds with probability $L/(m + L - 1)$. This is much worse than the $\approx m^{-1/L}$ winning probability that the optimal strategy will achieve. 

A better strategy would be for Alice to focus on repeatedly doubling her money. On each step, if she has $X_{t - 1}$ dollars, then she turns this into $2X_{t - 1}$ dollars with probability $1 - 1/(L + 1)$ and into $-(L - 1)X_{t - 1}$ dollars with probability $1/(L + 1)$. Repeating this strategy for $\lg_2 m$ steps (and assuming for simplicity that $m$ is a power of $2$), Alice's overall probability of winning becomes 
$$(1 - 1/(L + 1))^{\lg_2 m} \approx e^{- (\lg_2 m) / (L + 1)} = m^{-(\lg_2 e) / (L + 1)}.$$

This strategy reveals an important lesson: Rather than trying to make a single bet that will win her the entire game, Alice should try to grow her pool slowly so that she can have more and more money to work with (and to leverage!) on future bets. 

This suggests that the optimal strategy is the following: Alice should pick some small $\epsilon > 0$ (such that $m$ is a power of $(1 + \epsilon)$), and then repeatedly make bets so that, if on step $t$ she is still playing, then 
$$X_t = \begin{cases}
(1 + \epsilon) X_{t - 1} & \text{with probability } \frac{L}{L + \epsilon}, \\
-(L - 1)X_{t - 1} & \text{with probability } \frac{\epsilon}{L + \epsilon}.
\end{cases}$$
Note that this is a fair bet, since it makes a profit of $\epsilon X_{t - 1}$ with probability $\frac{L}{L + \epsilon}$ and a loss of $LX_{t - 1}$ with probability $\frac{\epsilon}{L + \epsilon}$. We call this strategy the \defn{$\epsilon$-doubling strategy}. 

% Note that the $\epsilon$-doubling strategy can be viewed as a variation of the strategy outlined in the introduction, but modified so that Alice's bets are scale-invariant (aiming to achieve fixed multiplicative growth rather than an additive one on each step). This distinction is important only in so far as that it makes some of the exposition below cleaner. 

\paragraph{Analyzing the $\epsilon$-doubling strategy. } We now argue that, as $\epsilon \to 0$, the $\epsilon$-doubling strategy achieves a success probability of $m^{-1/L}$. 

\begin{lemma}
    In the $L$-Leverage Game, the $\epsilon$-doubling strategy achieves a success probability of
    \[
        \left(\frac{L}{L + \epsilon}\right)^{\log_{1 + \epsilon} m}
        = m^{-\log_{1 + \epsilon}\left(1 + \frac{\epsilon}{L}\right)}.
    \]
    In particular, as $\epsilon \to 0$, this approaches $m^{-1/L}$.
\end{lemma}
\begin{proof}
Let $T = \log_{1 + \epsilon} m$. Since we chose $\epsilon$ so that $m$ is a power of $(1 + \epsilon)$, the quantity $T$ is an integer. Under the $\epsilon$-doubling strategy, Alice's fortune is multiplied by $(1 + \epsilon)$ each time she wins, and she goes bankrupt immediately if she loses. Thus she reaches $m$ if and only if she wins $T$ consecutive bets.

Each individual bet is won with probability $\frac{L}{L + \epsilon}$, so the success probability is exactly 
\[
    \left(\frac{L}{L + \epsilon}\right)^T
    = (1 + \epsilon/L)^{-T}
    = (1 + \epsilon/L)^{-\log_{1 + \epsilon} m}
    = (1 + \epsilon / L)^{-(\log_{1 + \epsilon/L} m) \log_{1 + \epsilon} (1 + \epsilon/L)}
    = m^{-\log_{1 + \epsilon} (1 + \epsilon/L)}.
\]
Finally, as $\epsilon \to 0$, we have $\log_{1 + \epsilon}(1 + \epsilon/L) \to 1/L$, so the success probability approaches $m^{-1/L}$.
\end{proof}

\paragraph{Optimality of $\epsilon$-doubling. } Finally, we show that the $\epsilon$-doubling strategy is, indeed, optimal as $\epsilon \to 0$:

\begin{proposition}
Any strategy Alice can use in the $L$-Leverage Game wins with probability at most $m^{-1/L}$.
\label{prop:optimality}
\end{proposition}

We can prove Proposition \ref{prop:optimality} by a simple application of the OST, but where rather than examining the martingale $X_0, X_1, \ldots$, we examine the supermartingale $Y_0, Y_1, \ldots$, where $Y_t = {X_t^+}^{1/L}$, and where $X_t^+ = \max\{X_t, 0\}$.

Specifically, if we can show that $Y_0, Y_1, \ldots$ is a supermartingale, then we can apply the OST to conclude that  
$$\E[Y_\tau] \le Y_0 = 1^{1/L} = 1.$$
But, depending on whether Alice wins, $Y_\tau$ is either $m^{1/L}$ or $0$, so $$\E[Y_\tau] =  \Pr[X_\tau = m] \cdot m^{1/L}.$$
This implies that $\Pr[X_\tau = m] \le 1/m^{1/L}$, as desired. 

To complete the proof of Proposition \ref{prop:optimality}, it suffices to show that $Y_0, Y_1, \ldots$ is, in fact, a supermartingale. The basic idea behind the proof is to show that, if there were any bet Alice could make that would make \begin{equation}\E[Y_t \mid X_0, X_1, \ldots, X_{t - 1}] > Y_{t - 1},\label{eq:breakmartingale}\end{equation}
then it would actually be the type of bet that the $\epsilon$-doubling strategy makes (for some $\epsilon > 0$), and it is straightforward to show that this type of bet does not result in \eqref{eq:breakmartingale}.
\begin{lemma}
The sequence $Y_0, Y_1, \ldots$ is a supermartingale.
\end{lemma}
\begin{proof}
Condition on the outcomes of $X_1, \ldots, X_{t - 1}$ with $X_{t - 1} > 0$. We wish to show that the conditional expectation of $(X_t^+)^{1/L}$ is at most $X_{t - 1}^{1/L}$. For $a \le X_{t - 1} \le b$, let $R_{a, b}$ be the random variable that has support $\{a, b\}$ and expected value $X_{t - 1}$. Since $\E[X_t] = X_{t - 1}$ and since the support of $X_t$ is contained in $[-(L - 1) X_{t - 1}, \infty)$, we can think of $X_t$ as a weighted combination of random variables from $\{R_{a, b} \mid -(L - 1) X_{t - 1} \le a < X_{t - 1} \le b\}$. It therefore suffices to focus on one such random variable $R_{a, b}$ and show that
 $$\E[(R_{a, b}^+)^{1/L}] \le X_{t - 1}^{1/L}.$$

Notice that, if $a, b \ge 0$, then 
$\E[(R_{a, b}^+)^{1/L}] = \E[R_{a, b}^{1/L}]$, which by Jensen's inequality (since $x \mapsto x^{1/L}$ is concave for $L \ge 1$) is at most $\E[R_{a, b}]^{1/L} =  X_{t - 1}^{1/L}$. 

If $a < 0$, then $\E[(R_{a, b}^+)^{1/L}]$ is maximized by setting $a$ to be as negative as possible, i.e., $a = -(L - 1)X_{t - 1}$. Moreover, 
$$\E[(R_{a, b}^+)^{1/L}] = \Pr[R_{a, b} = b] \cdot b^{1/L}.$$
Notice that, if we define $\epsilon > 0$ such that $b = (1 + \epsilon)X_{t - 1}$, then $R_{a, b}$ is exactly the type of bet the $\epsilon$-doubling strategy makes! Specifically, we have $\Pr[R_{a, b} = b] = \frac{L}{L + \epsilon}$, so 
$$\E[(R_{a, b}^+)^{1/L}] = (X_{t - 1} (1 + \epsilon))^{1/L}  \cdot \frac{L}{L + \epsilon} = X_{t - 1}^{1/L} (1 + \epsilon)^{1/L} / (1 + \epsilon/L).$$
Using the identity $(1 + x)^{1/L} \le 1 + x/L$ for $x \ge 0$, we obtain
$$\E[(R_{a, b}^+)^{1/L}] \le X_{t - 1}^{1/L},$$
as desired.
\end{proof}

This completes the proof that $\epsilon$-doubling, as $\epsilon$ goes to $0$, is an optimal strategy for the $L$-Leverage Game.

\paragraph{A remark on interest. }We remark that there are many ways in which the model in this paper is overly simplistic. One that is worth at least a brief discussion is the role of interest. A highly leveraged company such as Lehman Brothers does not get to borrow money \emph{for free}. Lehman estimated that it paid an annualized interest rate of approximately 5\% on its borrowed assets \cite[Pages 105--107]{Lehman2007AnnualReport}. 

One can naturally extend the L-Leverage Game to include interest, by requiring the company to, as part of each bet, pay $\delta \cdot k$ dollars interest, where $k$ is the amount of money that the company must borrow to place the bet. It should be noted, however, that so long as $\delta < 1/L$, this does not change the high-level conclusion of the analysis. The company can, for example, make repeated bets in which (prior to the interest payment) it either triples its money, or replaces its current amount $n$ with $-Ln$ dollars (going bankrupt). If the bet wins, then the interest payment is less than $n$ dollars, meaning the company still doubled its money. Each bet is won with probability $1 - O(1/L)$, so the probability of getting from $1$ dollar to $m$ dollars (winning $\log m$ bets in a row) remains $(1 - O(1/L))^{\log m} = m^{-O(1/L)}$.

\section*{Acknowledgments}
This work was supported in part by NSF grants CCF-2542165 and CCF-2504471, and by a Jane Street grant.

\bibliographystyle{plain}
\bibliography{refs}
\end{document}